\pdfoutput=1
\documentclass[a4paper]{lipics-v2021}

\usepackage{mathrsfs}
\usepackage{stmaryrd}
\usepackage{mathtools}
\usepackage{faktor}
\usepackage{ebproof}
\usepackage{hhline}
\usepackage{dsfont}
\usepackage{float}
\usepackage{xcolor}

\nolinenumbers

\newcommand{\set}[1]{\{#1\}}

\renewcommand{\implies}{\to}
\renewcommand{\iff}{\leftrightarrow}
\newcommand{\N}{\mathbb{N}}
\newcommand{\B}{\mathbb{B}}
\newcommand{\prop}{\mathit{Prop}}
\newcommand{\tmax}{\mathbf{Max}}
\newcommand{\eng}[1]{\left\langle#1\right\rangle}
\newcommand{\engexists}[1]{\left\langle#1\right\rangle^{\circ}}
\newcommand{\engexistsomega}[1]{\left\langle#1\right\rangle^{\circ}_{\N}}
\newcommand{\engomega}[1]{\left\langle#1\right\rangle_{\N}}
\newcommand{\res}[1]{\left\lfloor#1\right\rfloor}
\newcommand{\chain}{\mathbf{Ch}}
\newcommand{\UI}{\mathbf{UI}}
\newcommand{\GUI}{\mathbf{GUI}}
\newcommand{\GDC}{\mathbf{GDC}}
\newcommand{\AC}{\mathbf{AC}}
\newcommand{\FC}{\mathbf{FC}}
\newcommand{\OPEN}{\mathbf{OPEN}}
\renewcommand{\part}{\mathcal{P}}

\renewcommand{\Pi}{\prod}
\renewcommand{\Sigma}{\sum}
\newcommand{\Ind}{\mathbf{Ind}}
\renewcommand{\hat}[1]{\eng{#1}}
\newcommand{\EMPCF}{\mathbf{\exists MPCF}}
\newcommand{\LTT}{\mathbf{TTL}}
\newcommand{\LTTomega}{\mathbf{TTL}^{\N}}
\newcommand{\Zorn}{\mathbf{Zorn}}
\newcommand{\1}{\mathds{1}}

\newcommand{\Subchain}{\mathbf{SCh}}

\newcommand{\dom}{\mathbf{dom}}
\newcommand{\tmaxpf}{\mathbf{Max}_{\mathbf{dpf}}}
\newcommand{\tmaxrpf}{\mathbf{Max}_{\mathbf{rpf}}}
\newcommand{\G}{\mathcal{G}}
\newcommand{\inl}{\mathrm{inl}}
\newcommand{\inr}{\mathrm{inr}}
\newcommand{\impliescl}{\implies_{\mathrm{cl}}}
\newcommand{\parrow}{\to_\mathrm{p}}
\newcommand{\EMPCFm}{\EMPCF^-}
\newcommand{\LC}{\mathbf{C}}

\author{Hugo Herbelin}{Université de Paris Cité, Inria, CNRS, IRIF}{}{}{}
\author{Jad Koleilat}{Université Paris Cité}{}{}{}

\authorrunning{H. Herbelin, J. Koleilat}
\titlerunning{On the logical structure of some maximality and well-foundedness principles}

\Copyright{H. Herbelin, J. Koleilat}



\EventEditors{Jakob Rehof}
\EventNoEds{1}
\EventLongTitle{9th International Conference on Formal Structures for Computation and Deduction (FSCD 2024)}
\EventShortTitle{FSCD 2024}
\EventAcronym{FSCD}
\EventYear{2024}
\EventDate{July 10-13, 2024}
\EventLocation{Tallinn, Estonia}
\EventLogo{}
\SeriesVolume{299}
\ArticleNo{23}

\begin{document}

\title{On the logical structure of some maximality and well-foundedness principles equivalent to choice principles}

\maketitle

\begin{abstract}
We study the logical structure of \emph{Teichmüller-Tukey lemma}, a
maximality principle equivalent to the axiom of choice and show that
it corresponds to the generalisation to arbitrary cardinals of
\emph{update induction}, a well-foundedness
principle from constructive mathematics classically equivalent to
the axiom of dependent choice.

From there, we state general forms of maximality and well-foundedness
principles equivalent to the axiom of choice, including a variant of
Zorn's lemma. A comparison with the general class of choice
and bar induction principles given by Brede and the first author is
initiated.

\end{abstract}

\section{Introduction}

\subsection{Context} 

The axiom of choice is independent of Zermelo-Fraenkel set theory and
equivalent to many other
formulations~\cite{HorstHerrlich,Jech73,RubinRubin}, the most famous
ones being Zorn's lemma, a maximality statement, and Zermelo's
theorem, a well-ordering thus also well-foundedness theorem, since
well-foundedness and well-ordering are logically dual notions.

In the family of maximality theorems equivalent to the axiom of choice
one statement happens to be particularly concise and general, it is
Teichmüller-Tukey lemma, that states that every non-empty collection
of \emph{finite character}, that is, characterised only by its finite
sets, has a maximal element with respect to inclusion.

The axiom of dependent choice is a strict consequence of the
axiom of choice. In the context of constructive mathematics, various
statements classically but non intuitionistically equivalent to the
axiom of dependent choice are considered, such as bar induction, open
induction~\cite{Coquand92}, or, more recently, update induction~\cite{Open_Induction},
the last two relying on a notion of \emph{open} predicate over functions of
countable support expressing that the predicate depends only on finite
approximations of the function.

In a first part of the paper, we reason intuitionistically and show
that the notion of finite character, when specialised to countable
sets, is dual to the notion of open predicate, or, alternatively, that
the notion of open predicate, when generalised to arbitrary cardinals
is dual to the notion of finite character. As a consequence, we
establish that
update induction and the specialisation of Teichmüller-Tukey lemma to
countable sets are logically dual statements, or,
alternatively, that Teichmüller-Tukey lemma and the generalisation of
update induction to arbitrary cardinals are logically dual.

In a second part of the paper, we show how Teichmüller-Tukey lemma and
Zorn's lemma can be seen as mutual instances the one of the other.

Finally, in a third part, we introduce a slight variant of
Teichmüller-Tukey lemma referring to functions rather than sets and
make some connections with the classification of choice and bar
induction principles studied by Brede and the first author in~\cite{Herbelin1}.

The ideas of Section 2 have been developed during an undergraduate
internship of the second author under the supervision of the first
author in 2022, leading to the idea in Section 4.1 of introducing
$\EMPCF$ by the second author. Section 3 contains extra investigations
made in 2023 by the second author. Section 4.2 contains investigations
made jointly in 2024 by the authors.

\subsection{The logical system} 

In this section we describe the logical setting and give definitions that are used throughout the article. The results we prove do not depend greatly on its structure as they require only basic constructions, we shall make precise exactly was is necessary and what is left to the preferences of the reader.\\

We work in an intuitionistic higher order arithmetic equipped with inductive types like the type with one element ($\1$, $0 : \1$), the type of Boolean values ($\B$, $0_{\B}, 1_{\B}:\B$), the type of natural numbers ($\N$), the product type ($A \times B$), or the coproduct type ($A+B$). In particular, we write $B_\bot$ for the coproduct of $B$ and of $\1$, identifying $b : B$ with $\inl(b) : B_\bot$ and $\bot$ with $\inr(0)$ where $\inl$ and $\inr$ are the two injections of the coproduct.

We write $\prop$ for the type of propositions. For all types $A$, the type $\part(A)$ denotes the type $A \to \prop$, we shall sometimes refer to it as ``subsets of $A$''. We also use the type $\N \to A_\bot$, shortly $A_\bot^\N$, to represent the countable subsets of $A$, implicitly referring to the non-$\bot$ elements of the image of the function\footnote{For inhabited $A$, this is intuitionistically equivalent to considering $\N \to A$.}.

We also require a type for lists:  for all types $A$ we denote by $A^*$ the type of lists of terms of type $A$ defined as follows:
\begin{equation*}
\begin{prooftree}[center=false]
\hypo{}
\infer1{\varepsilon : A^*}
\end{prooftree}
\qquad 
\begin{prooftree}[center=false]
\hypo{u : A^* \quad a : A}
\infer1{u@a : A^*}
\end{prooftree}
\end{equation*}
We inductively define $\star : A^* \to A^* \to A^*$, the concatenation of two lists:
\begin{equation*}
\begin{prooftree}[center=false]
\hypo{u : A^*}
\infer1{u \star \varepsilon := u}
\end{prooftree}
\qquad 
\begin{prooftree}[center=false]
\hypo{u : A^* \quad  v : A^* \quad a : A}
\infer1{u \star (v@a) := (u \star v)@a}
\end{prooftree}
\end{equation*}
We denote by $[a_1, \dots, a_n]$ the list $(\dots(\varepsilon @ a_1) @ \dots ) @ a_n)$, since $\star$ is associative we drop the parentheses. If $n\in\N$ and $\alpha:A^\N$, we write $\alpha_{|n}$ for the recursively defined list $[\alpha(0), \ldots, \alpha(n-1)]$. We define $\in \ : A \to A^* \to \prop$ as: $a \in u := \exists v,w^{A^*}, \ v \star [a] \star w = u$. \\

The symbol $\in$ will be used as defined above and also as a notation for $P(a)$. To be more precise, for all types $A$, $P : \part(A)$ and $a : A$ we will write $a \in P$ for $P(a)$ and $a \notin P$ for $P(a) \implies \bot$. Continuing with the set-like notations, for $P, Q : \part(A)$ we write $P \subseteq Q$ for $\forall {a ^A}, a \in P \implies a \in Q$. We require extensional equality for predicates: for all $P,Q : \part(A)$, $P = Q \iff P \subseteq Q \land Q \subseteq P$\footnote{Extensionality for predicates is assumed for convenience, it is not fundamentally needed}. The symbol $\subseteq$ will also be used for lists: for all $u, v : A^*$, $u \subseteq v := \forall {a^A}, a \in u \implies a \in v$. Note that equipped with this relation, lists behave more like finite sets than lists. Nevertheless the list structure is not superfluous as will be shown later.

As a convention, we let the scope of quantifiers spans until the end of the sentence, so, for instance, $\forall n, P \implies Q$ reads as $\forall n, (P \implies Q)$ and similarly for $\exists$.

\subsection{Closure operators and partial functions}
\label{sec:closure}

Let us now define some closure operators and relations on subsets and lists:

\begin{definition} Let $A$ be a type, $u: A^*$, $\alpha : \part(A)$, $T:  \part(A^*)$, $P : \part(\part(A))$
\begin{align*}
&u \subset \alpha \ : \prop                                            &&\eng{T} \ :  \part(\part(A)) \\
&u \subset \alpha\ := \forall {a^A}, a \in u \to a \in \alpha &&\eng{T} \ :=  \lambda \alpha^{\part(A)}. \forall {u^{A^*}}, u \subset \alpha \to u \in T \\
&                             &&\engexists{T} \ :  \part(\part(A)) \\
&                             &&\engexists{T} \ := \lambda \alpha^{\part(A)}. \exists {u^{A^*}}, u \subset \alpha \land u \in T \\
&\hat{u} \ :  \part(A)                                     &&\res{P} \ : \part(A^*)\\
&\hat{u} \ :=  \lambda x^A. \ x \in u                      &&\res{P} \ := \lambda u^{A^*}. \ \hat{u} \in P
\end{align*}
\end{definition}

The symbol $\eng{ ~ }$ is the translation from ``the list world'' to ``the predicate world''. More precisely, $\hat{u}$ is the canonical way to see a list as a predicate ($u \subset \alpha \iff \hat{u} \subseteq \alpha$) and $\eng{T}$ is an extension of $T$ as a predicate on subsets, $\alpha : \part(A)$ is in $\eng{T}$ if and only if it can be arbitrarily approximated by lists of $T$. Dually, $\res{ ~ }$ is the translation from predicate to list, taking predicate of finite domain to all lists of elements in the domain. Note that $\eng{T}$ is downward closed, that is, $\alpha \subset \beta$ and $\beta \in \eng{T}$ implies $\alpha \in \eng{T}$. Note also that $\eng{\res{P}}$ is a downward closure operator, defining the largest downward closed subset of $P$. On its side, $\res{\eng{T}}$ builds the downward closure up to permutation and replication of the elements of the lists of $T$. Also, symmetrical properties applies to $\engexists{~}$ exchanging downward with upward and largest subset with smallest superset. Finally, notice that $\eng{T}$ may be empty, in fact $\eng{T}$ is inhabited if and only if $\varepsilon \in T$, and the same for $\engexists{T}$.

\paragraph*{Examples: }
Consider $T : \part(\B^*)$, for simplicity let us use set-like notations when defining $T$. If $T := \set{[1_\B,0_\B], [1_\B], [0_\B], \epsilon}$ then $\eng{T}$ will contain all subsets of $\B$. Now, if $T := \set{[1_\B,0_\B], [1_\B], [0_\B]}$, $\eng{T}$ will be empty since for all $\alpha : \part(\B)$, $\epsilon \subset \alpha$ but $\epsilon \notin T$. If $T := \set{\epsilon, [1_\B], [1_\B, 0_\B]}$ then $\eng{T}$ will contain only the empty subset and the singleton containing $1_\B$. Now consider $T' := \set{\epsilon, [1_\B], [1_\B, 1_\B], [0_\B, 1_\B], [1_\B, 0_\B, 1_\B, 1_\B]}$, notice that $\eng{T} = \eng{T'}$. The $\eng{ ~}$ operation does not care for duplications or permutations.\\
For $T := \set{ \epsilon, [1_\B],[1_\B, 0_\B]}$, $\res{\eng{T}}$ is $\set{\epsilon, [1_\B], [1_\B, 1_\B], [1_\B, 1_\B, 1_\B], \dots}$. Similarly, for $T := \set{ \epsilon, [1_\B], [0_\B], [1_\B, 0_\B]}$, $\res{\eng{T}}$ is the set of all lists on $\B$.\\
The $\engexists{ ~}$ operator has the dual behaviour. Consider $T : \part(\N^*)$, $T := \set{[1]}$ then, $\engexists{T}$ contains exactly all subsets of $\N$ containing $1$. Similarly if $\epsilon \in T$, then $\engexists{T}$ contains all subsets of $\N$. For $T := \set{[1]}$, $\res{\engexists{T}}$ will contain every list on $\N$ that contains at least one 1.\\

We also give similar definitions relatively to countable subsets, abbreviating $(A_\bot)^*$ into~$A_\bot^*$:

\begin{definition} Let $A$ be a type, $u: A_\bot^*$, $\alpha : A_\bot^\N$ and $T:  \part(A_\bot^*)$
\begin{align*}
&u \subset_\N \alpha \ : \prop                               &&\engomega{T} \ :  \part(A_\bot^\N) \\
&u \subset_\N \alpha\ := \exists {n^\N}, u = \alpha_{|n} &&\engomega{T} \ := \lambda \alpha^{A_\bot^\N}. \forall {u^{A^*}}, u \subset_{\N} \alpha \to u \in T \\
&                             &&\engexistsomega{T} \ :  \part(A_\bot^{\N}) \\
&                             &&\engexistsomega{T} \ := \lambda \alpha^{A_\bot^\N}. \exists {u^{A^*}}, u \subset_{\N} \alpha \land u \in T
\end{align*}
\end{definition}

We conclude this section defining two different notions of partial functions:

\begin{definition}[Relational partial function]
Let $A,B$ be types, a relational partial function $f$ from $A$ to $B$ is a relational functional relation of $\part(A \times B)$. Formally, a relational partial function from $A$ to $B$ is a term $f : \part(A \times B)$ such that $\forall a^A, \forall b,b'^B, \ ((a,b) \in f \land (a,b') \in f) \implies b = b'$. Its domain is defined by:
\begin{align*}
\dom(f) &: \part(A) \\
\dom(f) &:= \lambda a^A. \exists b^B, (a,b) \in f
\end{align*}
For all $a' : A$, we denote by $\dom(f) \cup a'$ the predicate $\lambda a^A. (\exists b^B, (a,b) \in f ) \lor a = a'$.
\end{definition}

\begin{definition}[Decidable partial function]
Let $A,B$ be types, a decidable partial function $f$ from $A$ to $B$ is a total function $f : A \to B_\bot$. Its domain and graph are defined by:
\begin{align*}
&\dom(f) : \part(A)                                           &&\G(f) : \part(A \times B) \\
&\dom(f) := \lambda a^A. f(a) \neq \bot        &&\G(f) := \lambda (a,b)^{A \times B}. f(a) = \inl(b)
\end{align*}
For all $a' : A$, we denote by $\dom(f) \cup a'$ the predicate $\lambda a^A. \ f(a) \neq \bot \lor a = a'$.
\end{definition}

\paragraph*{Notation:} We write $f \in A \parrow B$ to denote that $f$ is a relational partial function from $A$ to $B$ and $f : A \to B_\bot$ for the type of decidable partial functions from $A$ to $B$. We will also write $\Theta f^{A \parrow B}, P$ for $\Theta f^{\part(A \times B)}, (f \in A \parrow B) \implies P$ for $\Theta \in \set{\lambda, \forall, \exists}$. \\

The difference between these two definitions is in the decidability of the domain: decidable partial functions have a decidable domain while it's not the case of relational partial functions. The graph operation $\G$ allows us to recover a relational partial function from a decidable partial function. One needs excluded middle to recover a decidable partial function from a relation partial function, hence decidable partial functions are stronger axiomatically. Notice that we used the same notation $\dom$ in both definitions. Since they both have the same semantic meaning and we will make clear whether we are using relation partial function or decidable partial function, it should not cause any confusion.

\section{$\LTT$ and $\UI$}

In this section, we define Teichmüller-Tukey lemma and update induction and emphasise that they are logically dual, up to the difference that the former is relative to predicates over subsets of arbitrary cardinals while update induction is relative to predicates over countable subsets. Underneath, they rely on the dual notions of predicate of finite character and of open predicate.

\subsection{Predicates of finite character} 

A set is of \emph{finite character} if all its information is contained in its finite elements. In our setting, a predicate $P : \part(\part(A))$ is of finite character if all its information is contained in a predicate over lists. There are two canonical ways to express this:

\begin{definition}[Finite character]
Let $A$ be a type and $P : \part(\part(A))$. We propose two definitions of finite character:
\begin{align*}
P \in \FC_1 :=&  \ \forall {\alpha^{\part(A)}},  \alpha \in P \ \iff \  \forall {u^{A^*}},  u \subset \alpha \to u \in \res{P}\\
\nonumber \\
P \in \FC_2 :=&  \ \exists {T^{\part(A^*)}},  \eng{T} = P
\end{align*}
\\
Rewriting $\FC_1$ using the terms just defined:
\begin{equation*}
P \in \FC_1 := P = \eng{\res {P}}
\end{equation*}
\end{definition}

$\FC_1$ and $\FC_2$ are, in essence, paraphrases of one an other, thus there is no reason not to expect them to be equivalent. First we will need a lemma:

\begin{lemma}
\label{lemma_fcl}
Let $A$ be a type and $T : \part(A^*)$ then $\eng{T} \in \FC_1$.
\end{lemma}

\begin{proof}
Let $\alpha : \part(A)$. Suppose $\alpha \in \eng{T}$, our goal is to show that $\alpha \in \eng{\res{\eng{T}}}$. Let $u : A^*$ such that $u \subset \alpha$, we will show that $u \in \res{\eng{T}}$. By definition $u \in \res{\eng{T}}$ if and only if $\hat{u} \in \eng{T}$ if and only if every sublist of $u$ is in $T$. Since $\alpha$ can be arbitrarily approximated by terms of $T$ and $u \subset \alpha$, so can $u$. Hence, $u \in \res{\eng{T}}$ thus, $\alpha \in \eng{\res{\eng{T}}}$.\\
Suppose $\alpha \in \eng{\res{\eng{T}}}$, then for all $u : A^*$ such that $u \subset \alpha$, $u \in \res{\eng{T}}$ which we can rewrite as $\hat{u} \in \eng{T}$. We easily show that $\hat{u} \in \eng{T} \implies u \in T$ thus $\alpha \in \eng{T}$.
\end{proof}

We have shown that $\eng{T} = \eng{\res { \eng{T}}}$. This means that without loss of generality, we can require in $\FC_2$ that the witness $T$ be of the form $\res { \eng{T'}}$ for some $T'$. This is a way to express that $T$ can be chosen to be minimal. In fact if we are given $P$ and $T$ such as in $\FC_2$, it may happen that $T$ contains a list $u$ that is not closed under $\subseteq$ (i.e.. $v \subseteq u \not \implies v \in T$). Such an $u$ will be invisible when looking at $\eng{T}$, hence we can consider $u$ as a superfluous term. The $\res { \eng{ ~ }}$ operation allows us, without loss of generality, to remove those terms, thus making $T$ minimal. 

\begin{theorem}
\label{th_FC}
$\FC_1 \ \iff \ \FC_2$
\end{theorem}

\begin{proof}
Let $A$ be a type and $P : \part(\part(A))$. From left to right: suppose $P \in \FC_1$. $\res{P}$ is a witness of $P \in \FC_2$.\\
From right to left: suppose $P \in \FC_2$, let $T$ be the witness of $P \in \FC_2$. By lemma \ref{lemma_fcl} $\eng{\res{\eng{T}}} = \eng{T}$ and by hypothesis $P = \eng{T}$, we can rewrite the first equality as $\eng{\res{P}} = P$.
\end{proof}
Since $\FC_1$ and $\FC_2$ are equivalent, we will from now on write $\FC$ without the indices.

\subsection{Open predicates}

A notion of \emph{open} predicates over functions of countable domain
was defined in Coquand~\cite{Coquand92} and generalised by
Berger~\cite{Open_Induction}. Using the definitions of
Section~\ref{sec:closure}, a predicate over $\alpha:A^{\N}$ is
\emph{open} in the sense of Berger if it has the form $\alpha \in
\engomega{T} \implies \alpha \in \engexistsomega{U}$ for some $T, U
: \part(A^*)$. In order to get a closer correspondence with the notion
of finite character, we will however stick to Coquand's
definition. Additionally, to get a closer correspondence with the case
of open predicates used in update induction, we consider open
predicates for functions to $A_\bot$.

\begin{definition}[Countably-open predicate, in Coquand's sense, with partiality]
Let $A$ be a type and $P : \part(A_\bot^{\N})$. We define:
\begin{align*}
P \in \OPEN_{\N} :=& \  \exists T^{\part(A_\bot^*)},  \engexistsomega{T} = P
\end{align*}
\end{definition}

The observations made on predicates of finite character apply to
countably-open predicates, namely that $\engexistsomega{T} =
\engexistsomega{\res{\engexistsomega{T}}}$. Obviously, we can also
move from $A_\bot^{\N}$ to $\part(A)$ and introduce a general notion of
open predicates which again, will satisfy $\engexists{T} =
\engexists{\res{\engexists{T}}}$:

\begin{definition}[Open predicate]
Let $A$ be a type and $P : \part(\part(A))$. We define:
\begin{align*}
P \in \OPEN :=& \ \exists T^{\part(A^*)},  \engexists{T} = P
\end{align*}
\end{definition}

Conversely, we can define a notion of predicate of countably-finite
character dual the notion of countably-open predicate:

\begin{definition}[Predicate of countably-finite character]
Let $A$ be a type and $P : \part(A_\bot^{\N})$. We define:
\begin{align*}
P \in \FC_{\N} :=& \  \exists T^{\part(A_\bot^*)},  \engomega{T} = P
\end{align*}
\end{definition}

This finally results in
the following dualities:

\begin{table}[H]
  \caption{Predicates of finite character VS Open predicate}
  \label{table:def:opposite}
\renewcommand*{\arraystretch}{1.4}
\begin{center}
\begin{tabular}{|p{3cm}|p{4.5cm}|p{4.5cm}|}
\hline
 & \hfil Universal notion  & \hfil Existential notion\\
\hline
\hfil Arbitrary subsets & \hfil Finite character  & \hfil Open\\
\hfil Countable subsets & \hfil Countably-finite character  & \hfil Countably-open  \\
\hline
\end{tabular}
\end{center}
\end{table}

\subsection{Teichmüller-Tukey lemma and Update induction} 
\label{Teichmüller-Tukey lemma and Update induction}

Before defining Teichmüller-Tukey lemma we need a few definitions:

\begin{definition}
Let $A$ be a type, $P : \part(\part(A))$ and $\alpha, \beta : \part(A)$. We define:
\begin{align*}
\beta \prec \alpha  &:  \prop \\
\beta \prec \alpha &:= \exists {a^A}, a \notin \alpha \land \beta = (\lambda x^A. \ x \in \alpha \lor x = a) \\ 
\alpha \in \tmax_\prec(P)  &:  \prop \\
\alpha \in \tmax_\prec(P)  &:= \alpha \in P \land \forall {\beta^{\part(A)}},  \beta \prec \alpha \implies \beta \notin P
\end{align*}
\end{definition}

Thus, $\beta \prec \alpha$ stands for $\beta$ extends $\alpha$ (if $\beta$ is an update of $\alpha$) while $\tmax_\prec(P)$ is the predicate of maximal elements of $(P, \succ)$ ($\succ$ is the reverse of $\prec$).\\

What we are interested in are predicates of finite character but Theorem \ref{th_FC} allows us to consider only predicates on lists since there is a correspondence between them. Hence, we will quantify or instantiate schemas on predicate on lists.

\begin{definition}[Teichmüller-Tukey lemma]
Let $A$ be a type and $T: \part(A^*)$, we define the Teichmüller-Tukey lemma, $\LTT_{AT}$:
\begin{equation*}
(\exists {\alpha^{\part(A)}}, \alpha \in \eng{T}) \implies \exists {\alpha^{\part(A)}},  \ \alpha \in \tmax_\prec(\eng{T})
\end{equation*}
\end{definition}

\paragraph*{Notations:} $\LTT$ denotes the full schema: for all types $A$ and all $T : \part(A^*)$, \ $\exists {\alpha^{\part(A)}}, \alpha \in \eng{T} \implies \exists {\alpha^{\part(A)}}, \alpha \in \tmax_\prec(\eng{T})$.\\
$\LTT_{AT}$ denotes the schema specialised in this $A$ and this $T$.\\
$\overline{\LTT_{AT}}$ denotes the restriction of the full schema $\LTT$ to $A$ and $T$ of a particular shape. For example: $\overline{\LTT_{\N T}}$ is the schema: for all $T : \part(\N^*)$, \ $\exists {\alpha^{\part(\N)}}, \alpha \in \eng{T} \implies \exists {\alpha^{\part(\N)}}, \alpha \in \tmax_\prec(\eng{T})$. Moreover, if $C_A$ denotes a particular collection of predicates over lists of $A$ ($A$ is a parameter), then $\overline{\LTT_{A C_A}}$ denotes the restrictions of the schema $\LTT$ to any $A$ type and $T :  \part(A^*)$ that is in $C_A$.\\

Following an earlier remark, we impose that the finite character predicate we are considering must be inhabited, without this $\LTT$ becomes trivially inconsistent. Having defined $\LTT$ we now show that we can recover an induction principle by using contraposition and Morgan's rules:\\
\\
Unfolding some definitions, $\LTT_{AT}$ is
\begin{equation*}
(\exists { \alpha^{\part(A)}}, \alpha \in \eng{T}) \implies \exists { \alpha^{\part(A)}}, \alpha \in \eng{T} \land (\forall { \beta^{\part(A)}},  \beta \prec \alpha \implies \beta \notin \eng{T})
\end{equation*}
Contraposing and pushing some negations:
\begin{equation*}
\forall { \alpha^{\part(A)}}, [ \neg (\alpha \in \eng{T}) \lor \neg \forall { \beta^{\part(A)}},  \beta \prec \alpha \implies \beta \notin \eng{T} ] \implies \forall { \alpha^{\part(A)}}, \alpha \notin \eng{T}
\end{equation*}
We have a sub-formula of the form $\neg A \lor \neg B$, we have the choice of writing it as $A \implies \neg B$ or $B \implies \neg A$. The first choice leads to a principle we will call $\LTT^{\mathrm{co}}_{AT}$:
\begin{equation*}
\forall { \alpha^{\part(A)}}, [ \alpha \in \eng{T} \implies \exists { \beta^{\part(A)}},  \beta \prec \alpha \land \beta \in \eng{T} ] \implies \forall { \alpha^{\part(A)}}, \alpha \notin \eng{T}
\end{equation*}
And the second choice leads to an induction principle: 
\begin{equation*}
\forall { \alpha^{\part(A)}}, [( \forall { \beta^{\part(A)}},  \beta \prec \alpha \implies \beta \notin \eng{T}) \implies \alpha \notin \eng{T} ] \implies \forall { \alpha^{\part(A)}}, \alpha \notin \eng{T}
\end{equation*}

$\LTT^{\mathrm{co}}$ is intuitively an opposite formulation of $\LTT$. The induction principle we obtain seems to express something different. We can push further the negations in order to obtain a positive formulation of it:
\begin{equation*}
\forall { \alpha^{\part(A)}}, [( \forall { \beta^{\part(A)}},  \beta \prec \alpha \implies \beta \in \engexists{T}) \implies \alpha \in \engexists{T} ] \implies \forall { \alpha^{\part(A)}}, \alpha \in \engexists{T}
\end{equation*}

And this can be seen as as a generalisation of Berger's update
induction~\cite{Open_Induction} going from countably-open predicates
to arbitrary open predicates.

To state update induction, we need to focus on partial functions from
$\N$ to $A$ which we equip with an order:

\begin{definition}
Let $A$ be a type, $P : \part(A_\bot^{\N})$ and $\alpha, \beta : A_\bot^{\N}$. We define:
\begin{align*}
\beta \prec_N \alpha  &:  \prop \\
\beta \prec_N \alpha &:= \exists {m^\N}, \exists {a^A}, \alpha(m) = \bot \land \beta(m) = a \land \forall n^\N, n \neq m \implies \alpha(n) = \beta(n)
\end{align*}
\end{definition}

Like $\LTT$, update induction is originally defined on open predicates but since any open predicate comes from a predicate on lists, we can more primitively state it as follows:

\begin{definition}[Update induction]
Let $A$ be a type and $T: \part(A_\bot^*)$, we define Update induction, $\UI_{AT}$:
\begin{equation*}
\forall { \alpha^{A_\bot^{\N}}}, [( \forall { \beta^{A_\bot^{\N}}},  \beta \prec_N \alpha \implies \beta \in \engexistsomega{T}) \implies \alpha \in \engexistsomega{T} ] \implies \forall { \alpha^{A_\bot^{\N}}}, \alpha \in \engexistsomega{T}
\end{equation*}
\end{definition}

Contrastingly, we now formally state the dual of $\LTT$ that we obtained above:

\begin{definition}[Generalised update induction]
Let $A$ be a type and $T: \part(A^*)$, we define Generalised update induction, $\GUI_{AT}$:
\begin{equation*}
\forall { \alpha^{\part(A)}}, [( \forall { \beta^{\part(A)}},  \beta \prec \alpha \implies \beta \in \engexists{T}) \implies \alpha \in \engexists{T} ] \implies \forall { \alpha^{\part(A)}}, \alpha \in \engexists{T}
\end{equation*}
\end{definition}

Also, we introduce a countable version of $\LTT$, logically dual to $\UI$:

\begin{definition}[Countable Teichmüller-Tukey lemma]
Let $A$ be a type and $T: \part(A_\bot^*)$, we define the countable Teichmüller-Tukey lemma, $\LTTomega_{AT}$:
\begin{equation*}
(\exists {\alpha^{A_\bot^{\N}}}, \alpha \in \engomega{T}) \implies \exists {\alpha^{A_\bot^{\N}}}, \alpha \in \tmax_{\prec_\N}(\engomega{T})\\
\end{equation*}
\end{definition}

We thus obtain the following table:

\begin{table}[H]
  \caption{Maximality principles VS Induction principles}
  \label{table:def2}
\renewcommand*{\arraystretch}{1.4}
\begin{center}
\begin{tabular}{|p{3cm}|p{4cm}|p{4cm}|}
\hline
 & \hfil Finite character  & \hfil Open \\
\hline
\hfil Arbitrary subsets & \hfil $\LTT_{AT}$ & \hfil $\GUI_{AT}$ \\
\hfil Countable subsets & \hfil $\LTTomega_{AT}$ & \hfil $\UI_{AT}$ \\
\hline
\end{tabular}
\end{center}
\end{table}

In particular, since $\LTT$ is classically equivalent to the full
axiom of choice, $\GUI$ is also classically equivalent to the full
axiom of choice.

\section{$\LTT$ and Zorn's lemma} 

In this section we analyse precisely the relationships of $\LTT$ with Zorn's lemma. We go further than showing their equivalence, we look at which part of $\LTT$ (as a schema) is necessary to prove Zorn's lemma and reciprocally. This equivalence result is also a proof that the version of Teichmüller-Tukey lemma we defined captures the full choice.

\begin{definition}
Let $A$ be a type, $<$ a strict order on $A$, $a : A$ and $E, F : \part(A)$. Define: 
\begin{align*}
E \in \chain(A) &: \prop\\                                                                                       
E \in \chain(A) &:=  \forall {a,b^A},  a, b \in E \implies (a < b \lor b < a \lor a = b)\\ \\
F \in \Subchain(E) &: \prop\\
F \in \Subchain(E) &: F \subseteq E \land  F \in \chain(A)\\ \\
E \in \Ind(A) &:  \prop\\                               
E \in \Ind(A) &:= (\forall {F^{\part(A)}}, F \in \Subchain(E) \implies \exists {a^A}, a \in E \land \forall {b^A},  b \in F \implies b \leq a)  \\ \\
a \in \tmax_<(E)  &:  \prop \\
a \in \tmax_<(E)  &:= a \in E \land \forall {b^A},  a < b \implies b \notin E
\end{align*}
Where $\leq$ is the reflexive closure of $<$. 
\end{definition}

$\chain$ is the chain predicate, $\Subchain$ is the subchain predicate, $\Ind$ is the inductive ``subset'' predicate and $\tmax_<$ is simply the maximal element predicate. We choose to express these definitions in terms of predicates over types rather than directly in terms of types, to avoid discussions on proof relevance and stay in a more general setting. If we were proof-irrelevant, instantiating our schemas on predicates over types would be identical to doing it directly on types which would simplify notations and yield the same results.\\
We can now define concisely Zorn's lemma:

\begin{definition}[Zorn lemma] Let $A$ be a type, $<$ a strict order on $A$, and $E$ a predicate on $A$. $\Zorn_{A<E}$ is the following statement
\begin{equation*}
E \in \Ind \implies \exists {a^A}, a \in \tmax_<(E)
\end{equation*}
\end{definition}

\begin{theorem}
$\LTT \iff \Zorn$
\end{theorem}

The following is an adaptation of a usual set-theoretic proof in our setting.

\begin{proof} From left to right: fix $A$ a type, $<$ a strict order on $A$ and $E : \part(A)$ such that $E \in \Ind(A)$. We first show that $\Subchain(E)$ is of finite character:\\
\\
Let $F : \part(A)$ such that $F \in \Subchain(E)$, we show $F \in \eng{ \res{ \Subchain(E)}}$: let $u : A^*$ such that $u \subset F$, $\hat{u}$ is thus a chain of $E$ therefore $u \in \res{\chain(E)}$. Let $F : \part(A)$ such that $F \in \eng{ \res{ \Subchain(E)}}$, by choosing lists of length 2 we can show that $F$ is a subchain of $E$. Hence $\Subchain(E) \in \FC$.\\
\\
Using $\LTT_{A \res{\Subchain(E)}}$, we get $G : \part(A)$ such that $G \in \tmax(\Subchain(E))$. $G$ is a subchain of $E$, since $E$ is inductive we get $g : A$ such that $g \in E$ and $\forall {aA}, a \in G \implies a < g$. Suppose we have $h : A$ such that $g < h$ and $h \in E$ . Let $G' := \lambda a^A. a \in G \lor a = h$, then we have $G' \prec G$, since $G \in \tmax(\Subchain(E))$, $G' \notin \Subchain(E)$. On the other side, $G'$ is obviously a chain and $G' \subseteq E$, therefore $G' \in \Subchain(E)$. This is a contradiction, hence $g \in \tmax_<(E)$.\\

From right to left: let $T : \part(A^*)$. $\subset$ is a strict order on $\part(\part(A))$. Since $\eng{T}$ is of finite character, a maximal element for $\subset$ is also a maximal element for $\succ$. Hence, what is left to show is that $\eng{T}$ is inductive and use $\Zorn_{\part(A) \subset \eng{T}}$ to produce a maximal term. Let $Q : \part(\part(A))$ such that $Q \in \Subchain(\eng{T})$. Let $\alpha := \lambda a^A. \exists {\beta^{\part(A)}}, \beta \in Q \land a \in \beta$. By construction, $\alpha$ is an upper bound of $Q$, let's show that it is indeed in $\eng{T}$. Since $\eng{T}$ is of finite character it suffices to show that for all $u : A^*$, $u \subset \alpha \implies u \in T$. Let $u : A^*$ such that $u \subset \alpha$. Since $u$ is a finite list, we can enumerate its elements $a_0, \dots, a_n$. For all $0 \leq i \leq n$, let $\beta_i : \part(A)$ be such that $a_i \in \beta_i $ and $\beta_i \in Q$. Since $Q$ is chain, there exists $0 \leq i_0 \leq n$ such that for all $0 \leq i \leq n, \beta_i \subseteq \beta_{i_0}$. Thus, $u \subset \beta_{i_0}$, $\beta_{i_0} \in \eng{T}$ and so $u \in \eng{T}$.
\end{proof}

Looking more closely at this proof we notice that we have proved a finer result than simply the equivalence. We have shown $\LTT_{A\res{\Subchain(E)}} \implies \Zorn_{A<E}$ and $\Zorn_{\part(A) \subset \eng{T}} \implies \LTT_{AT}$. We can express for a predicate over lists to be of the form $\res{\Subchain(E)}$ in a more syntactic way.

\begin{definition}
Let $A$ be a type and $T : \part(A^*)$, we say that $T$ is a list of chains, if there exists $T'$ such that:
\begin{itemize}
\item $\epsilon \in T'$
\item  $u@a \in T'$ and $[a] \star v \in T'$ if and only if  $u \star [a] \star v \in T'$
\item $u \star [a] \star v \in T'$ implies $u \star v \in T'$
\item if $a \neq b$ and $u \star [a] \star v \star [b] \star w \in T'$ then for all $u', v', w' : A^*$, $u' \star [b] \star v' \star [a] \star w' \notin T'$ 
\end{itemize} 
and $T$ is the downward closure of $T'$ by $\subseteq$. We denote by $\LC_A$ the collection of lists of chains of $A$. 
\end{definition}

\begin{lemma}
Let $A$ be a type, $<$ a strict order on $A$ and $E : \part(A)$, then there exists $T \in \LC_A$ such that $\Subchain(E) = \eng{T}$. Reciprocally, let $A$ be a type, then for every $T \in \LC_A$ there exist a strict order $<$ on A and $E : \part(A)$ such that $\Subchain(E) = \eng{T}$. 
\end{lemma}

\begin{proof}
Proof of the first statement: we inductively define a $T' : \part(A^*)$. 
\begin{equation*}
\begin{prooftree}[center=false]
\hypo{}
\infer1{\varepsilon \in T'}
\end{prooftree}
\qquad
\begin{prooftree}[center=false]
\hypo{a \in E}
\infer1{[a] \in T'}
\end{prooftree}
\qquad
\begin{prooftree}[center=false]
\hypo{b \in E \quad a < b \quad u@a \in T'}
\infer1{u@a@b \in T'}
\end{prooftree}
\end{equation*}
We easily show that $T'$ satisfies the conditions of the above definition. Let $T$ be the downward closure of $T'$. Let $F \in \Subchain(E)$ and $u : A^*$ such that $u \subset F$. Since $F$ is a chain we can construct a list $u'$ of all elements of $u$ such that $u'$ does not contain twice the same element and is ordered increasingly relative to $<$. $u'$ is thus in $T'$ hence $u$ is in $T$. Let $F \in \eng{T}$ and $a,b : A$ such that $a,b \in F$. By hypothesis the list $[a,b]$ is in $T$. There exists $u \in T'$ such that $[a,b] \subset u$. Hence $a,b \in \eng{u}$ which is a chain. In conclusion $F$ is a subchain of $E$.\\

Proof of the reciprocal: suppose given a type $A$ with decidable equality and $T \in \LC_A$. There exists a $T'$ satisfying the aforementioned conditions. Let $E := \lambda a^A. \exists {u^{A^*}}, u \in T' \land a \in u$. We now must define an ordering on $A$. Define $<$ a binary relation on $A$ such that $a < b := [a,b] \in T'$. Using last "axiom" of the definition of $T'$ we easily show that it is irreflexive. For transitivity notice that if $[a,b],[b,c] \in T'$ then $[a,b,c] \in T'$ then $[a,c] \in T'$. Thus, it is a strict ordering on $A$. Let $F \in \Subchain(E)$ and $u : A^*$ such that $u \subset F$. We can assume that $u$ is sorted increasingly relatively to $<$. Using the same trick used for proving transitivity show that $u \in T$. Let $F \in \eng{T}$ and $a,b : A$ such that $a,b \in F$. By hypothesis the list $[a,b]$ is in $T$ therefore, $a < b$ which means that $F$ is indeed a chain. 
\end{proof}

\begin{corollary}
$\overline{\LTT_{A\LC_A}} \implies \Zorn$ and $\overline{\Zorn_{\part(A) \subset \eng{T}}} \implies \LTT$. Hence we deduce the somewhat surprising results $\LTT \iff \overline{\LTT_{A\LC_A}}$ and $ \Zorn \iff \overline{\Zorn_{ \part(A) \subset \eng{T}}}$. 
\end{corollary}

Looking back at the path we took to arrive at this conclusion, the results are quite expected, but looking only at the definition of a list of chains it is quite surprising that restricting $\LTT$ this much does not change its power.

\section{$\EMPCF$} 

In this section we define a choice principle $\EMPCF$ which stands for ``Exists a Maximal Partial Choice Function'' and a weaker version $\EMPCFm$. It is weaker in the sense that $\EMPCF$ implies $\EMPCFm$ but the equivalence is true if we allow excluded middle. We show that  $\EMPCFm$ is equivalent in its general form to $\LTT$ and link $\EMPCF$ to the general class of dependent choice $\GDC$, given by Brede and the first author in \cite{Herbelin1}. In particular, $\EMPCF$ and $\EMPCFm$ can be seen as refinements of $\LTT$ whose strength is more explicitly controlled.
  
\begin{definition}
Let $A,B$ be types, $f,g \in A \parrow B$ and $P : \part(\part(A \times B))$, define:
\begin{align*}
g \prec f &: \prop \\
g \prec f &: \exists {a^A}, a \notin \dom(f) \land (\dom(g) = \dom(f) \cup a )\  \land \\ &(\forall x^A, x \in \dom(f) \implies \exists b^B, (x,b) \in f \land (x,b) \in g)\\
f \in \tmaxrpf (P) &: \prop \\
f \in \tmaxrpf (P) &:=  f \in P \land \forall {g^{A \parrow B}}, g \prec f \implies g \notin P
\end{align*}
\end{definition}

\begin{definition}[$\EMPCFm$] 
Let $A, B$ be types and $T : \part((A \times B)^*)$, $\EMPCFm_{ABT}$ is the statement:
\begin{equation*}
(\exists \alpha^{\part( A \times B)}, \alpha \in \eng{T}) \implies \exists {f^{A \parrow B}}, f \in \tmaxrpf(\eng{T})
\end{equation*}
\end{definition}

\begin{definition}
Let $A, B$ be types, $f,g : A \to B_\bot$ and $P : \part(\part(A \times B))$, define:
\begin{align*}
g \prec f &: \prop \\
g \prec f &: \exists {a^A}, a \notin \dom(f) \land (\dom(g) = \dom(f) \cup a) \ \land \\ &(\forall x^A, x \in \dom(f) \implies f(x) = g(x)) \\
f \in \tmaxpf (P) &: \prop \\
f \in \tmaxpf (P) &:=  \G(f) \in P \land \forall {g^{A \to B_\bot}}, g \prec f \implies \G(g) \notin P
\end{align*}
\end{definition}

Since the intuitive meaning is the same we use the symbol $\prec$ for predicate, for relational partial functions and decidable partial function.

\begin{definition}[$\EMPCF$] 
Let $A, B$ be types and $T : \part((A \times B)^*)$, the theorem of existence of a maximal partial choice function $\EMPCF_{ABT}$ is the statement:
\begin{equation*}
(\exists \alpha^{\part( A \times B)}, \alpha \in \eng{T}) \implies \exists {f^{A \to B_\bot}}, f \in \tmaxpf(\eng{T})
\end{equation*}
\end{definition}

The difference between $\EMPCF$ and $\EMPCFm$ lies solely in the "kind" of partial function that is used. Hence, as per the above remark on the differences between relation partial function and decidable partial function, $\EMPCF \implies \EMPCFm$ and assuming excluded middle $\EMPCFm \implies \EMPCF$ which we denote by $\EMPCFm \impliescl \EMPCF$.

\subsection{$\EMPCF$ and $\LTT$} 

Now that we have defined $\EMPCFm$, we show that it is equivalent to $\LTT$ hence, $\EMPCF \implies \LTT$ and $\LTT \impliescl \EMPCF$.

\begin{theorem}
Let $A$ be a type, $T : \part(A^*)$ and $\pi^*T$ the operation that maps $T$ to $\lambda u^{(A \times \1)^*}. \ \pi(u) \in T$ where $\pi$ is the canonical projection of $(A \times \1)^*$ on $A^*$. Then,\\  $\EMPCFm_{A\1 \pi^*T}  \implies \LTT_{AT}$. Let $A$, $B$ be types and $T : \part((A \times B)^*)$ then, $\LTT_{(A \times B) T} \implies \EMPCFm_{ABT}$.
\end{theorem}

\begin{proof} $\EMPCFm_{A\1 \pi^*T}  \implies \LTT_{AT}$: let $A$ a type, $T : \part(A^*)$ and $\pi^*T := \lambda u^{(A \times \1)^*}. \ \pi(u) \in T$. From $\EMPCFm_{A\1\pi^*T}$ we obtain $f \in A \parrow \1$ such that $f \in \tmaxrpf(\eng{\pi^*T})$. Define $\alpha := \dom(f)$ and let's show that $\alpha \in \tmax(\eng{T})$. By construction, $\alpha$ is in $\eng{T}$. Suppose $\beta : \part(A \times B)$ such that $\beta \prec \alpha$. We can construct a relational partial function $g : A \parrow \1$ such that $\beta = \dom(g)$. Since $g \prec f$, $g$ is not in $\eng{U}$ hence $\beta$ is not in $\eng{T}$.\\
\\
$\LTT_{(A \times B) T} \implies \EMPCFm_{ABT}$: let $A,B$ types and $T : \part(( A \times B)^*)$. Define
\begin{equation*}
Q := \lambda u^{(A \times B)^*}. \ (\forall {a^A}, \forall {b,b'^B}, (a,b) \in u \land (a,b') \in u \implies b = b') \land u \in T
\end{equation*}
Notice that $\eng{Q}$ is not empty, since $\eng{T}$ is inhabited, $\epsilon \in T$. From this, we deduce that $\epsilon \in Q$ hence, the empty predicate is in $\eng{Q}$. We can now apply $\LTT_{(A\times B) Q}$ and get $\alpha$ such that $\alpha \in \tmax(\eng{Q})$. By construction $\alpha$ is a relational partial function. It follows that it's a maximal relational partial function, thus proving $\EMPCFm_{ABT}$.
\end{proof}

$\LTT$ can be seen as a projection of $\EMPCF$. The fact that they are so tightly linked is not surprising as ``being a partial function'' for a subset of $A \times B$ is a property of finite character. 

\subsection{$\EMPCF$ and $\GDC$} 

Introduced in~\cite{Herbelin1}, Generalised Dependent Choice ($\GDC_{ABT}$) is a common generalisation of the axiom of dependent choice and of the Boolean prime ideal theorem. Parameterised by a domain $A$, a codomain $B$ and a predicate $T : \part((A\times B)^*)$, it yields dependent choice when $A$ is countable, the Boolean prime ideal theorem when $B$ is two-valued, and the full axiom of choice when $T$ comes as the ``alignment'' of some relation (see below). To the difference of $\EMPCF$, $\GDC$ asserts the existence of a total choice function, but this to the extra condition of a property of ``approximability'' of $T$ by arbitrary long finite approximations. To the difference of $\EMPCF$ whose strength is the one of the full axiom of choice, expecting a total choice function makes $\GDC$ inconsistent in its full generality.

In this section we investigate how restricting $\EMPCF$ to countable $A$ or two-valued $B$ impacts its strength to exactly the same extent as it restricts the strength of $\GDC$. Two such preliminary results are given, but first, let's translate $\GDC$ in our setting:

\begin{definition}[$A$-$B$-approximable]
Let $A,B$ be types and $T : \part((A \times B)^*)$. For all $X :  \part((A \times B)^*)$ define 
\begin{equation*}
\phi(X) := \lambda u^{(A \times B)^*}.\, ( u \in \res{\eng{T}} \land \forall {a^A}, \neg (\exists {b^B}, \ (a,b) \in u)  \implies \exists {b^B},  u@(a,b) \in X )
\end{equation*}
The $A$-$B$-approximation of $T$ denoted $T_{AB \text{ap}}$ is the greatest fixed point of $\phi$. We say that $T$ is $A$-$B$-approximable if $\varepsilon \in T_{AB \text{ap}}$.
\end{definition}


\begin{definition}[$A$-$B$-choice function]
Let $A,B$ be types and $T : \part((A \times B)^*)$. $T$ has an $A$-$B$-choice function if:
\begin{equation*} 
\exists {f^{A \to B}}, \forall {u^{(A \times B)^*}},  u \subset \G(f) \implies u \in T
\end{equation*} 
\end{definition}

\begin{definition}[$\GDC$]
Let $A,B$ be types and $T : \part((A \times B)^*)$, $\GDC_{ABT}$ is the statement: if $T$ is $A$-$B$-approximable then $T$ has an $A$-$B$-choice function.
\end{definition}

\begin{theorem}
\label{GDC-EMPCF-N}
$\overline{\GDC_{\N B T}} \impliescl \overline{\EMPCF_{\N BT}}$
\end{theorem}

\begin{proof}
Let  $B$ be a type and $T : \part((\N \times B)^*)$. In order to use $\GDC$, $T$ must be $\N$-$B$-approximable but the $T$ we are given might not be. Thus, we are going to construct $T_\bot : \part((\N \times B_\bot)^*)$ that is $\N$-$B_\bot$-approximable and use $\GDC$ to obtain a function that we will prove maximal.\\
\\
For all $u : \part((A \times B_\bot)^*)$ define $\overline{u}$ inductively:
\begin{equation*}
\begin{prooftree}[center=false]
\hypo{}
\infer1{\overline{\varepsilon} := \varepsilon}
\end{prooftree}
\qquad
\begin{prooftree}[center=false]
\hypo{a : A \quad b : B }
\infer1{\overline{u@(a, b)} := \overline{u}@(a, b)}
\end{prooftree}
\qquad
\begin{prooftree}[center=false]
\hypo{a : A}
\infer1{\overline{u@(a, \bot)} := \overline{u}}
\end{prooftree}
\end{equation*}
By induction define $T^n_\bot : \part((\N \times B_\bot)^*)$:
\begin{itemize}
\item $T^0_\bot := \lambda u^{(\N \times B_\bot)^*}. \ u = \varepsilon$
\item Let $T^{n+1}_\bot$ be defined inductively
\begin{equation*}
\begin{prooftree}[center=false]
\hypo{u \in T_n \quad b : B \quad \overline{u}@(n+1,b) \in T}
\infer1{u@(n+1,b) \in T^{n+1}_\bot}
\end{prooftree}
\qquad
\begin{prooftree}[center=false]
\hypo{u \in T_n \quad \forall {b^B}, \overline{u}@(n+1,b) \notin T}
\infer1{u@(n+1, \bot) \in T^{n+1}_\bot}
\end{prooftree}
\end{equation*}
\end{itemize}
Now define $T_\bot$ as the $\subseteq$-downward closure of the union of the $T^{n}_\bot$. We must show that $T_\bot$ is  $\N$-$B_\bot$-approximable. By definition $T_\bot = \res{\eng{T_\bot}}$. Let $n : \N$, $v : (\N \times B_\bot)^*$ such that $v \in T_\bot$ and $\neg (\exists {c^{B_\bot}}, \ (n,c) \in v)$. By definition, there exists $m : \N$ and $u \in T^{m}_\bot$ such that $v \subseteq u$. If $n \leq m$ then there exists $c : B_\bot$ such that $(n,c) \in u$, thus $v@(n,c) \subseteq u$ and $v@(n,c) \in T_\bot$. If $n > m$ then there exists $u' \in T^n_\bot$ such that $u \subseteq u'$. It is in the proof of this statement that we need excluded middle to show that we always satisfy the hypothesis of one of the induction steps. Hence, $v \subseteq u'$ and we now repeat the same argument.  $T_\bot$ satisfies $\phi$ and contains $\varepsilon$, thus we can apply $\GDC_{\N B_\bot T_\bot}$ and get $f : \N \to B_\bot$ a choice function.\\
\\
What is left to show is that $f$ is a maximal partial function. Let $n : \N$ such that $n \notin \dom(f)$ and let $g : \N \to B_\bot$ extending $f$ with $\dom(g) = \dom(f) \cup n$. Let us write $f_{<n}$ for the list $[(0, f(0)), \dots, (n-1, f(n-1))]$. $f_{< n} \in T^{n}_\bot$ and since $f_{< n+1}$ is of the form $f_{< n}@(n, \bot)$ by case analysis we deduce that $\forall {b^B}, \overline{f_{< n}}@(n,b) \notin T$. If $\G(g) \in \eng{T_\bot}$ then $g_{< n+1} \in T_\bot$ and $g_{< n+1} = f_{< n}@(n, g(n))$ with $g(n) : B$.  $\overline{f_{< n}}@(n,g(n))$ is thus in $T$, contradiction. Hence, $f$ is maximal.
\end{proof}

Let's write $\mathbf{DC}$ for the axiom of dependent choice. We have:

\begin{corollary}
Since $\overline{\GDC_{\N B T}}$ is equivalent to $\mathbf{DC}$  \cite{Herbelin1} we deduce: $\mathbf{DC} \impliescl \overline{\LTT_{(\N \times B) T}}$
\end{corollary}

\begin{theorem}
For $A$ a type with decidable equality, $\overline{\EMPCF_{A \B T}} \implies \overline{\GDC_{ A \B T}}$
\end{theorem}

\begin{proof}
Let $A$ be a type and $T : \part((A \times \B)^*)$ $A$-$\B$-approximable. Define $U := \res{ \eng{ T_{A \B \text{ap}}}}$, the $A$-$\B$-approximable hypothesis guarantees that $\eng{U}
$ is inhabited. Using $\EMPCF_{A \B U}$ we get $f : A \to \B_\bot$ a maximal partial choice function. We show that $f$ must be total, that is that it is impossible that it takes the value $\bot$. Indeed assume $f(a) = \bot$ for some $a : A$ and consider $g : A \to \B_\bot$ that extends $f$ with $g_0(a) = 0_\B$. We have $g \prec f$, 
thus $\G(g) \notin \eng{U}$. Then, there exists $u : (A \times \B)^*$ such that $u \subset \G(g)$ and $u \notin 
U$. Using the decidability of equality in $A$, we can find $u'$ such that $u = u'@(a,0_\B)$ where $u' \subset \G(f)$. Symmetrically, by considering the extension $g$ of $f$ obtained by setting $g(a) = 1_\B$, there exists $v' \subset \G(f)$ 
such that $v'@(a,1_\B) \notin U$. Since $u' \star v' \subset \G(f)$, $u' \star v' \in U$. There must be $b : \B$ such that $(u' \star v')@(a,b) \in U$. But in both cases ($b = 0_\B$ or $1_\B$) 
there is a sublist ($u'@(a,0_\B)$ or $v'@(a,1_\B)$) that is not in $U$, contradiction. Hence, $f$ is total.
\end{proof}

The following definition, taken from~\cite{Herbelin1}, describes how
to turn a relation on $A$ and $B$ as a predicate over $(A \times
\B)^*$ that filters approximations.

\begin{definition}[Positive alignment]
Let $A$ and $B$ be types and $R$ a relation on $A$ and $B$. The
\emph{positive alignment $R_\top$ of $R$} is the predicate on $(A
\times B)^*$ defined by:
$$R_\top := \lambda u.\forall (a,b) \in u, R(a,b)$$
\end{definition}

Positive alignments can be characterised by the following property.

\begin{definition}[Downward prime]
Let $A$ and $B$ be types. We say that $T : \part((A \times B)^*)$ is
\emph{downward prime} when $u \in T$ and $v \in T$ implies $u \star v
\in T$. We denote by $\mathbf{D}_{AB}$ the collection of downward prime $T : \part((A \times B)^*)$.
\end{definition}


\begin{theorem}
If $R$ is a relation on $A$ and $B$, its positive alignment is
downward prime. Conversely, if $T$ is downward prime, it is the
positive alignment of the relation $|T|$ defined by
$$|T|(a,b) := [(a,b)] \in T$$
\end{theorem}

\begin{proof}
This is because $u \star v \in R_\top$, that is $\forall (a,b) \in u \star
v, R(a,b)$ is equivalent to $(\forall (a,b) \in u, R(a,b)) \land (\forall (a,b) \in
v, R(a,b))$, that is to $u \in R_\top \land v \in R_\top$, and, conversely, because $u \in |T|_\top$, that is
$\forall (a,b) \in u, [(a,b)] \in T$, is equivalent, by induction on $u$, using
downward primality at each step, to $u \in T$.
\end{proof}

Based on the equivalence between $\AC_{ABR}$ and $\GDC_{ABR_\top}$
in~\cite[Thm 7]{Herbelin1}, we obtain:
\begin{corollary}
$\GDC_{ABT}$ for $T$ downward prime characterises the full axiom of choice
  $\AC_{ABR}$, that is $\forall x^A, \exists y^B, R(a,b) \implies
  \exists f^{A \to B}, \forall x^A, R(a,f(a))$.
\end{corollary}


We now show that $\GDC_{ABT}$ is also equivalent to $\EMPCF_{ABT}$ for
$T$ downward prime.

\begin{theorem}
For $T: \part((A \times \B)^*)$ downward prime for $A$ with decidable
equality, $\EMPCF_{A B T} \implies \GDC_{A B T}$.
\end{theorem}

\begin{proof}
Since $T$ is $A$-$B$-approximable, it contains $\varepsilon$, so that
$\eng{T}$ is non-empty. Thus, by $\EMPCF_{A B T}$, we get $f : A \to
B_\bot$ a maximal partial choice function. We show that $f$ must be
total.  Indeed, assume $a : A$ such that $f(a) = \bot$. By
$A$-$B$-approximability, we can obtain a $b$ such that $[(a,b)] \in
\res{\eng{T}}$. Let's now consider the function $g : A \to B_\bot$
defined by setting $g(a') = b$ if $a = a'$ and $g(a') = f(a')$
otherwise.  We have $g \prec f$, thus $\G(g) \notin \eng{T}$. But this
contradicts that we can also prove that any $u \subset \G(g)$ is in
$T$, that is $\G(g) \in \eng{T}$. Indeed, by decidability of equality
on $A$, either $u$ has an element of the form $(a,b')$ or not. In the
second case, $u \subset \G(f)$ and thus $u \in T$. In the first case,
$u$ has the form $u' \star (a,b') \star u''$ with $u' \in \G(f)$ and
$u'' \in \G(f)$, thus $u' \in T$ and $u'' \in T$. Since $u \subset
\G(g)$, we also have $b' = g(a) = b$. Then, by downward primality, we
get $u' \star [(a,b)] \star u'' \in T$.
\end{proof}

\begin{theorem}
For $T: \part((A \times B)^*)$ downward prime, $\overline{\GDC_{A B \mathbf{D}_{AB}}} \implies \overline{\EMPCFm_{A B \mathbf{D}_{AB}}}$.
\end{theorem}

\begin{proof}
There are two ways to embed a partial function from $A$ to $B$ into a
total function: either restrict $A$ to the domain of the function, or
extend $B$ into $B_\bot$, as in Theorem~\ref{GDC-EMPCF-N}. We give a
proof using the first approach.

Let $A'$ be the subset of $A$ such that $\exists b^B, [(a,b)] \in T$.
We show coinductively that if $A'$ is infinite, the restriction of $T$
on $A'$ is $A'$-$B$-approximable. First, we do have $\varepsilon \in
T$ because $\eng{T}$ is non empty. Then, assume $u \in T$ and $a : A'$
such that $\neg (\exists {b^B}, \ (a,b) \in u)$ (which is possible
since $A'$ is supposed infinite). Since $a$ is in $A'$, there is $b$
such that $[(a,b)] \in T$, and by downward primality, $u \star (a,b)
\in T$, hence $A'$-$B$-approximable by coinduction.

Thus, there is a total function $f:A' \to B$ such that $\G(f) \in
\eng{T}$, which induces a partial function $f'$ from $A
\parrow B_\bot$. It remains to show that $f'$ is maximal.
Let $a \notin \dom(f)$, that is such that $\forall b^B, \neg [(a,b)] \in
T$. Then, there is obviously no extension of $f'$ on $a$ that would be
in $\eng{T}$.

It remains to treat the case of $A'$ finite, which can be obtained by
(artificially) reasoning on the disjoint sum of $A'$ and $\N$, and
setting $T[(n,p)] := (n=p)$ on $\N$.
\end{proof}

\section{Conclusion}

While Brede and the first author~\cite{Herbelin1} investigated the
general form of a variety of choice and bar induction principles seen
as contrapositive principles, this paper initiated the investigation
of a general form of maximality and well-foundedness principles
equivalent to the axiom of choice. One of the surprise was that, up to
logical duality, two principles such as Teichmüller-Tukey lemma and
Berger's update induction were actually of the very same nature. By
seeing all these principles as schemes, we could also investigate how
to express Zorn's lemma and Teichmüller-Tukey lemma as mutual
instances the one of the other. Finally, by starting investigating how
maximality, when applied to functions, relates to totality in
the presence of either a countable domain or a finite codomain, we
initiated a bridge between maximality and well-foundedness
principles and the general family of choice and bar induction principles
from~\cite{Herbelin1}.

The investigation could be continued in at least five directions:
\begin{itemize}
\item In the articulation between $\LTT$ and $\EMPCF$: assuming an
  alternative definition of $\LTT$, say $\LTT^+$, where $\part{(A)}$ is
  represented as a characteristic function from $A$ to $\B$, that is,
  equivalently, as a function from $A$ to $\1_\bot$, one would get the
  following identifications:

$$
\begin{array}{ccccccc}
\LTT^+_{AT} &=& \EMPCF_{A\1\pi^*T} & \qquad & \LTT^+_{(A\times B)T} &=& \EMPCF_{ABT}\\
\LTT_{AT} &=& \EMPCF^-_{A\1\pi^*T} & \qquad & \LTT_{(A\times B)T} &=& \EMPCF^-_{ABT}\\
\end{array}
$$

\newcommand{\DC}{\mathbf{DC}}
\newcommand{\GBI}{\mathbf{GBI}}

\item In the articulation between a sequential definition of
  countably-finite character and countably-open predicate, as in
  $\LTT^N_{BT}$ and $\UI_{BT}$, and a non-sequential definition, as in
  $\EMPCF_{\N BT} $ and $\EMPCFm_{\N BT}$, similar to the connection
  between $\DC^{\mathit{prod.}}_{BT}$ and $\GDC_{\N BT}$ in
  \cite{Herbelin1}.

\item In the relation between $\EMPCF_{A BT}$ and $\EMPCFm_{A BT}$
  on one side and $\GDC_{ABT}$ on the other side, verifying that the
  correspondences between $\EMPCF_{\N BT}$ and $\GDC_{\N BT}$, and
  between $\EMPCF_{A\B T}$ and $\GDC_{A\B T}$ hold, at least
  classically, in both directions, the same way as they do in the case
  $T$ downward prime.

\item In the articulation between $\LTT$ and $\GUI$, formulating statements
  dual to $\EMPCF$ and $\EMPCF-$ and connecting them to the dual of $\GDC$, that is $\GBI$
  \cite{Herbelin1}, analysing the role of classical reasoning and
  decidability of the equality on the domain in the correspondences.

\item In the relation between $\LTT$, $\EMPCF$, $\EMPCF-$ and other
  maximality principles than Zorn's lemma, also studying other
  well-foundedness principles than $\UI$.
\end{itemize}

In particular, an advantage of $\EMPCF$ and $\EMPCF-$ over $\GDC$ is
that their more general form is classically equivalent to the axiom of
choice while the most general form of $\GDC$ is inconsistent.

\bibliography{citations}

\end{document}